\documentclass{ws-ijmpb}

\usepackage{graphicx}
\usepackage{epstopdf} 
\usepackage{dcolumn} 
\usepackage{bm}
\usepackage{subfigure}
\usepackage{color}

\newcommand{\ket}[1]{| #1 \rangle}

\begin{document}

%
\catchline{}{}{}{}{}
%
\title{Quantum discord in the ground state of spin chains}

\author{MARCELO S. SARANDY}

\address{Instituto de F\'isica, Universidade Federal Fluminense, Av. Gal. Milton
Tavares de Souza s/n, Gragoat\'a, 24210-346, Niter\'oi, Rio de Janeiro, Brazil \\
msarandy@if.uff.br}

\author{THIAGO R. DE OLIVEIRA}

\address{Instituto de F\'isica, Universidade Federal Fluminense, Av. Gal. Milton 
Tavares de Souza s/n, Gragoat\'a, 24210-346, Niter\'oi, Rio de Janeiro, Brazil}

\author{LUIGI AMICO}

\address{ CNR-MATIS-IMM $\&$ Dipartimento di Fisica e Astronomia Universit\`a di Catania, C/O ed. 10, viale A. Doria 6,
  95125 Catania, Italy }

\maketitle


\begin{abstract}
The ground state of a quantum spin chain is a natural playground for investigating correlations. 
Nevertheless, not all correlations are genuinely of quantum nature. Here we review the recent progress 
to quantify the 'quantumness' of the correlations throughout the phase diagram of quantum spin 
systems. Focusing to one spatial dimension,  we discuss the behavior of quantum discord close to 
quantum phase transitions. In contrast to the two-spin entanglement,  pairwise discord 
is effectively long-ranged in critical regimes. Besides the features of quantum phase transitions, quantum discord is especially feasible to 
explore the factorization phenomenon, giving rise to  nontrivial ground classical states in quantum 
systems. The effects of spontaneous symmetry breaking are also discussed as well as the identification 
of quantum critical points through correlation witnesses. 
\end{abstract}


\section{Introduction}
Quantum mechanics relies on properties of complex vector spaces, whose elements are associated with wavefunctions. 
Therefore any correlation exploiting the quantum superposition principle (expressing the notion of 
 sum in vector spaces) is indeed of quantum nature. This simple evidence has been exploited as a 
guiding principle to quantify  quantum correlations beyond the generic notion of 'correlation' in a 
condensed matter system. In this context,  quantum spin systems naturally provide quantum correlations, 
especially at $T=0$, where the system is in its ground state. Such correlations are diversely important throughout 
the phase diagram of the system. For example, entanglement, the most famous quantum correlation, 
turned out strong enough close to  quantum phase transitions (QPT)~\cite{Sachdev:book,Continentino:book}, 
but it  can be very small deep in the phase of a highly correlated quantum spin system  where quantum 
fluctuations provide a factorized ground state ~\cite{Kurman:82,Verrucchi:05,Illuminati:08} (see also 
Ref.~\refcite{Amico:08} for a review on entanglement in many-body systems).

A recently devised measure of quantum correlation is the {\it quantum discord} (QD)~\cite{Ollivier:01}, 
providing a  'quantum/classical sieve'  explicitely  based on  the superposition principle of quantum 
mechanics~\cite{Ollivier:01,Vedral:01,Modi:12}. With QD, it was evidenced that indeed there are quantum 
correlations in separable mixed states. This, in turn,  may disclose  new scenarios for the phase diagram 
of  many-body quantum systems. The first computations of QD for spin systems in the thermodynamical limit 
have been carried out by Dillenschneider~\cite{Dillenschneider:08} and Sarandy~\cite{Sarandy:09}, where QD 
between two spins in a mixed state realized by the rest of the system was considered. It turned out that the 
non-analyticities of energy derivatives at QPTs were reflected not only in the two-spin entanglement~\cite {Wu:04} 
and QD~\cite{Dillenschneider:08,Sarandy:09} but also in the classical correlations, establishing all them as useful 
tools to analyze quantum critical phenomena. Further investigations were then performed in a number of many-body 
systems (see, e.g., Refs.~\cite{Chen:10,Maziero:10,Sun:10,Allegra:11,Liu:11,Li:11,Bin:11,Pal:11,Cheng:12,Wang:12,Chen:12}), 
where  entanglement, QD and classical correlations were compared. The bulk of evidence we can rely on so far 
suggests that genuine quantum correlations, including entanglement and QD, exhibit the behavior expected by 
scaling theory for the long-range physics close to QPTs. For the short-range physics, in contrast, a rather 
complex situation emerges because entanglement, QD and standard correlation functions may vanish for different 
ranges~\cite{Osterloh:02,Maziero:12} \footnote{This, in turn,  might be important at low temperature, because 
indeed thermal fluctuations act as infrared cutoff in the system (the QD close to QPT at low temperature is 
reviewed in Ref.~\refcite{Werlang:12} of this focus issue; see also Refs.~\cite{Amico07,Sun:11,Jiang:11,Hassan:10,Li:12}).}. 
This is clearly evidenced at factorizing fields occurring in the ordered phase of the system (characterized by a nonvanishing 
spectroscopic gap). Close to factorization, the entanglement range diverges despite the ordinary correlation functions 
decay exponentially\cite{Amico06}. 

In this article we review the properties of QD throughout the phase diagram of spin systems in a one-dimensional lattice, 
with nearest-neighbor exchange interactions. For these systems, different ground states may be achieved through 
order-disorder QPTs  occurring via the spontaneous symmetry breaking (SSB) mechanism. We consider both the so 
called  {\it thermal states} and the {\it symmetry broken ground states}, enjoying and breaking the symmetries of the Hamiltonian, 
respectively. As long as the system is studied through its energy spectrum,  all thermodynamical observables stay 
unaltered irrespective whether the thermal ground state or the state with broken symmetry are considered. 
Nevertheless, the difference comes in entanglement and QD  (the same argument can be applied to any other measure of quantum 
correlation based on reduced density matrix). Entanglement was analyzed in the SSB scenario in the 
Refs.~\cite{Osborne:02,Olav,Osterloh,Oliveira}. For the analysis of QD, SSB effects have been first accounted in 
Refs.~\cite{Tomasello:11,Tomasello:12}. The article is organized as follows. In section \ref{section_discord} we introduce 
QD. The models we deal with are introduced in the section \ref{section_models}; in there we provide a summary of the 
symmetry breaking mechanism to make the article self contained. In the sections \ref{section_thermal} and \ref{section_broken} 
we review the properties of the QD for thermal and symmetry broken  ground states respectively.  As illustration, 
in section \ref{section_witness}  we discuss a witness for multipartite quantum correlations~\cite{Saguia:11}, for both thermal 
and symmetry broken ground states. In  section \ref{section_conclusions} we draw our conclusions.

\section{Quantum discord}
\label{section_discord}
Let us begin by considering a composite system in a bipartite Hilbert space 
${\cal H} = {\cal H}_{A} \otimes {\cal H}_{B}$. The system is characterized 
by quantum states described by density operators $\rho \in {\cal B}({\cal H})$, 
where ${\cal B}({\cal H})$ is the set of bound, positive-semidefinite operators 
acting on ${\cal H}$ with trace given by ${\textrm{Tr}}\,\rho=1$.
Then, denoting by $\rho$ the density matrix of the composite system $AB$ and by 
$\rho^A$ and $\rho^B$ the density matrices of parts $A$ and $B$, respectively, 
the total (classical + quantum) correlations between parts $A$ and $B$ can be 
provided by the quantum mutual information~\cite{Groisman:05} 
\begin{equation}
I(\rho^A : \rho^B) = S(\rho^A) - S(\rho^A | \rho^B),
\label{qmi1}
\end{equation}
where $S(\rho^A) = -{\textrm{Tr}} \rho^A \log \rho^A$ is the von Neumann entropy for subsystem $A$  
(with the symbol log denoting logarithm at base 2) and
\begin{equation} 
S(\rho^A | \rho^B) = S(\rho) - S(\rho^B)
\label{qce1}
\end{equation}
is the quantum conditional entropy for $A$ given the part $B$.
A remarkable observation 
realized in Ref.~\refcite{Ollivier:01} is that the conditional entropy can be introduced by a different 
approach, which yields a result in the quantum case that differs from Eq.~(\ref{qce1}). Indeed, 
let us consider a measurement performed locally only on part $B$. This measurement can be described 
by a set of projectors $\{B_k\}$. The state of the quantum system, conditioned on the measurement of 
the outcome labelled by $k$, becomes
\begin{equation}
\rho_k = \frac{1}{p_k} \left( I\otimes B_k \right) \rho \left( I\otimes B_k \right), 
\label{rhok}
\end{equation}
where $p_k = {\textrm{Tr}} [ (I\otimes B_k ) \rho ( I\otimes B_k ) ]$ denotes the probability of 
obtaining the outcome $k$ and $I$ denotes the identity operator for the subsystem $A$. The conditional 
density operator given by Eq.~(\ref{rhok}) allows for the following measurement-based definition of 
the quantum conditional entropy:
\begin{equation}
S(\rho|\{B_k\}) = \sum_k p_k S(\rho_k).
\end{equation}
Therefore, the quantum mutual information can also be alternatively defined by
\begin{equation}
J(\rho:\{B_k\}) = S(\rho^A) - S(\rho|\{B_k\}).
\label{qmi2}
\end{equation}
Eqs.~(\ref{qmi1}) and~(\ref{qmi2}) are classically equivalent but they are different in the quantum case. The 
difference between them is due to quantum effects on the correlation between parts $A$ and $B$ and provides  
a measure for the quantumness of the correlation. Indeed, by generalizing the measure set $\{B_k\}$ to a  
positive-operator valued measure (POVM) and maximizing over such POVMs, we can define the classical correlation 
between parts $A$ and $B$ as~\cite{Vedral:01}
\begin{equation}
C(\rho) = \max_{\{B_k\}} J(\rho:\{B_k\}),
\label{ccorrel}
\end{equation}
Then, quantum correlations can be accounted by subtracting Eq.~(\ref{ccorrel}) from Eq.~(\ref{qmi1}), 
which yields
\begin{equation}
Q(\rho) = \min_{\{B_k\}} \left[ I(\rho^A : \rho^B) - J(\rho:\{B_k\}) \right] .
\label{qcorrel}
\end{equation}
Quantum correlations as measured by $Q(\rho)$ define QD. As provided by 
Eq.~(\ref{qcorrel}), it is an asymmetric quantity that measures the quantum correlations between 
subsystems $A$ and $B$ revealed by measurements on part $B$. Then, if QD vanishes, the system is in 
a {\it quantum-classical} state. In order to define {\it classical-classical} states, we can symmetrize 
QD with respect to measurements on each subsystem~\cite{Maziero-2:10,Rulli:11}, requiring the vanishing 
of QD for having total classicality.

Here we apply the above notions to the situation where $A$ and $B$ are individual spins 
of a quantum spin chain. To compute QD, we will adopt the specific case of von Neumann 
local measurements, which are provided by the {\it orthogonal} projectors $\hat{B}_k$, such that 
\begin{equation}
\hat{B}_{k} = V \hat{\Pi}_{k}V^{\dag} \quad,\quad k=0,1
\label{projectors}
\end{equation}
where $\{\hat{\Pi}_{k} \} = \{ | k \rangle \langle k| \}$ define the computational basis and $V \in SU(2)$ is a 
rotation operator. Then, $V$ can be parametrized as 
\begin{equation}
  V = \begin{pmatrix} \cos \frac{\theta}{2} &  \sin \frac{\theta}{2} \, e^{-i\phi} \\[1.1ex]
    - \sin \frac{\theta}{2} \, e^{i\phi} & \cos \frac{\theta}{2}
  \end{pmatrix} \, ,
  \label{parametr_V}
\end{equation}
where $\theta \in[0,\pi]$ and $\phi \in[0,2\pi)$ are respectively the azimuthal and polar axes 
of a qubit over in the  Bloch sphere~\cite{Dillenschneider:08,Sarandy:09}. Hence, QD as given 
by Eq.~(\ref{qcorrel}) can be directly obtained by minimizing over all angles 
$\theta$ and $\phi$. For two spin-1/2 systems, it is shown that the optimal POVM in 
Eq.~\eqref{qcorrel} is a projective measurement (4-projector-elements POVMs)~\cite{dariano,Hamieh:04}. 
Remarkably (see Refs.~\cite{Ali:10,Lu:11,Chen:11}), for the particular case of Bell-diagonal states 
(marginal states are maximally mixed), such optimal POVM can be achieved through von Neumann measurements, 
while for the less restricted case of $Z_2$-symmetric states (X-states), there are examples where optimization 
over POVMs may be inequivalent to von Neumann measurements. In such a situation and also for generic nonsymmetric 
states, which is the case of states with SSB, orthogonal measurements are not more effective than 4-elements 
POVMs~\cite{zaraket}, with von Neumann measurements in Eq.~\eqref{qcorrel} providing just upper bounds for QD 
(as an illustration, see Ref.~\refcite{Amico:12}). 
\section{The models}
\label{section_models}

Magnetic interactions constitute a simple mechanism to produce classical and quantum correlations. 
In this context, we consider the $XYZ$ chain, which is composed by $N$ spin-$1/2$ particles localized 
in a one dimensional lattice, interacting anisotropically along the three spatial directions with 
Heisenberg interaction strengths $J_x, \, J_y, \, J_z$, and subjected to a uniform external field $h$, 
which is governed by the Hamiltonian
\begin{equation}
\label{general-spin}
{\cal {H}} {=}  \frac{1}{2}\sum_{i,j} 
\left[ J_x S_i^x S_{j}^x + J_y S_i^y S_{j}^y + J_z S_i^z S_{j}^z \right]-
h \sum_{i} S_i^z \; ,
\end{equation}
where  $S_i^\alpha$ ($\alpha=x,y,z$) are spin-$1/2$ operators defined on the $i$-th site of the chain. 
Factorization of the ground state and QPTs characterize the phase diagram of the system.
The bulk of evidence at disposal so far indicates that the two phenomena are interconnected, with the 
factorization being a precursor of  QPTs~\cite{Amico:08}.

\subsection{QPTs and symmetry breaking}

QPTs occur when  an external perturbation (like pressure, magnetic field, etc.) causes quantum fluctuations 
that are effective enough to lead to a qualitative change of the ground state of the system. In the present 
article, we consider order-disorder QPTs within the symmetry breaking mechanism.  Accordingly, the ordered 
phase of the spin system is characterized by a finite magnetization, that is a local order parameter breaking 
the symmetry of the Hamiltonian. In the case of Eq. (\ref{general-spin}), the symmetry is a global $Z_2$ symmetry 
generated by  rotations of $\pi$ around the $z$ direction of all spins:
$
\sigma^{x}_j\rightarrow-\sigma^{x}_j = \Pi^{\dagger}\sigma^{x}_j\Pi $, 
$\sigma^{y}_j\rightarrow-\sigma^{y}_j = \Pi^{\dagger}\sigma^{y}_j\Pi \, , 
$ 
where  $\Pi = \bigotimes_{i=1}^{N} \sigma^z_i$. This indicates  that the ground states corresponding to opposite 
$ \langle gs|\sigma^{x,y}_j|gs\rangle$  are indeed degenerate and span a ground-state-manifold where any state is a 
possible ground state of the system. Because of such a symmetry  (together with the fact that $\Pi|gs\rangle=e^{i\phi}|gs\rangle$),  
the magnetization in the $xy$ vanishes identically in the ground state:
\begin{eqnarray}
\langle gs|\sigma^{x,y}_j|gs\rangle = \langle gs|\Pi^{\dagger}\sigma^{x,y}_j\Pi|gs\rangle
=-\langle gs|\sigma^{x,y}_j|gs\rangle\; .
\end{eqnarray}
The so called 'thermal ground state' is the ground state that is an equal mixture of the two degenerate ground states; 
it can be thought as the limit  $T\rightarrow0$ of the thermodynamical state. Thermal states enjoy the same symmetry of the 
Hamiltonian. Superpositions of degenerate ground states may also preserve the symmetry. However, they are demonstrated to be 
not stable to small perturbations because of  the lack of clustering property~\cite{Parisi}. Therefore the symmetry is 
spontaneously broken in physical systems (in the thermodynamical limit) because any small in-plane local magnetic field lifts 
the degeneracy and there is no observable 
connecting the two ground states (super-selection rule~\cite{super-selection}). For finite systems, for example, we could have a 
superposition of all the spins in opposed directions like $|\uparrow\uparrow...\uparrow\rangle\pm|\downarrow\downarrow...\downarrow\rangle$. 
These two states are orthogonal and preserve the symmetry, but they are not stable against small perturbations either 
(see also Ref.~\refcite{Rossignoli}). Here we will consider both the thermal ground states and the state with explicit symmetry 
breaking of the $Z_2$ symmetry. We observe that thermal ground state and the state with broken symmetry provide the same thermodynamics 
of the system (see~\cite{Sandvik} for a recent reference); the corresponding density matrices, however, differ sensibly. 

\subsection{Factorization of the ground state} 
The ground state of Eq.(\ref{general-spin}) is, in general, a highly entangled state. Nonetheless, there may exist points where 
it is indeed a product state, which is factorized as the tensor product of individual spin states~\cite{Kurman:82}. 
For the model Hamiltonian in Eq.~(\ref{general-spin}) the factorization occurs at~\footnote{It is remarkable that the 
factorization occurs for $d$-dimensional spin system on a bipartite lattice and for finite range exchange interaction 
even in presence of frustration~\cite{Verrucchi:05,Illuminati:08} as the field $h$ is varied across the factorizing 
field $h_{\rm f}$, yielding
\begin{equation}
h_f=\frac{1}{2}\sqrt{({\cal J}_x-{\cal J}_z)({\cal J}_y-{\cal J}_z) }\; ,
\end{equation}
where ${\cal J}_{x,y}=\sum_{|i-j|}^\infty (-)^{|i-j|}Z_{|i-j|} J_{x,y}^{|i-j|}$, and ${\cal J}_{x,y}=\sum_{|i-j|}^\infty Z_{|i-j|} J_{x,y}^{|i-j|}$, $Z_r$ indicating the range of interaction among the spins.}
\begin{equation}
h_f = 2 \sqrt{(J_x-J_z) (J_y-J_z)} \, .
\end{equation}
Indeed,  a qualitative change of the entanglement, called Entanglement Transition (ET),  occurs at $h_f$, involving two gapped phases and with a divergent entanglement range~\cite{fubini,Amico:06}. We also mention that analysis of the ground state for finite size systems demonstrated that the factorization  can be viewed as transition between ground states of different parities (see Ref.~\cite{giorgidepasquale} and and a related article in this special issue \cite{Canosa:12}). Close to the factorization, the system is characterized by the following pattern of correlations functions:
\begin{equation}
g^{\alpha,\beta}(r)= A^{\alpha,\beta}+B^{\alpha,\beta}(r) (h-h_f)+{\cal O}[(h-h_f)^2] \; ,
\label{corr_pattern}
\end{equation} 
where we remark that $A^{\alpha,\beta}$ is the same constant $\forall r$. Eq. (\ref{corr_pattern}) was obtained for the quantum XY 
model (case II below) in Ref.~{\refcite{baroni}}, but it holds for the generic model in Eq.(\ref{general-spin})~\cite{Quintanilla:unpubl}. 

We will discuss Eq.~\eqref{general-spin} in the following cases. 

\paragraph{(I)} The non-integrable antiferromagnetic $XYX$ model ($J_x=J_z=1$), with $J_y = 1/4$
(this is the case experimentally realized with ${\rm Cs_2 \, Co \, Cl_4}$~\cite{kenzelmann}). 
At  zero field $h$ the system is critical in the same universality class of the isotropic XY model (see 
Fig.~\ref{f3-1} for $\vert J_x \vert = 1$). At finite $h$ the system acquires a finite gap, vanishing at $h\simeq 3.21$ were a QPT occurs of the Ising type.
The order parameter is $\langle S^x \rangle $. The factorization is displayed at $h\simeq 3.16$.  
\begin{figure}
\begin{center}
\hspace*{-5cm}
\includegraphics[scale=0.5]{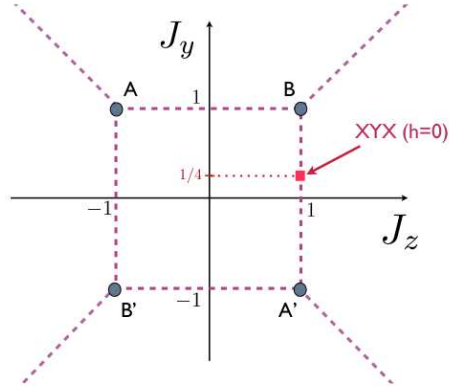}
\end{center}
\vspace*{-5cm}
\caption{Phase diagram of the integrable $XYZ$ model (Eq.(\ref{general-spin}) in zero field $h=0$). 
The dashed (red) lines correspond to critical phases. Everywhere else the system is gapped. In the points 
$A,A',B,B'$ the system is  tricritical.}
\label{f3-1}
\end{figure}

\paragraph{(II)}
For $J_z=0$ $J_x=J(1+\gamma), J_y=J (1-\gamma), J_z=0$, the (quantum anisotropic $XY$) Hamiltonian can be diagonalized by  first applying  the Jordan-Wigner transformation and then performing a Bogoliubov transformation~\cite{Lieb,Pfeuty,XYsol}. 
The  quantum Ising  model corresponds to  $\gamma =1$ while the (isotropic) $XX$-model is recovered for $\gamma = 0$. In the  latter (isotropic) case the model enjoys an additional symmetry resulting in the conservation of the total magnetization along the $z$-axis. The properties 
of the Hamiltonian are governed by  the dimensionless coupling constant $h/J$ (where $J $ is set as the energy scale and working in units such that $\hbar = 1$).
The phase diagram is sketched in Fig.~\ref{f3-2}~\cite{Takahashi-book}. In the interval  $0<\gamma\le 1$ the system undergoes a second order quantum phase transition at the critical value $h_c=1$. The order parameter is the magnetization in $x$-direction,  $\langle S^x\rangle $, different from zero for $h <1$. In the phase with broken 
symmetry the ground state has a two-fold degeneracy reflecting a global phase flip symmetry of the system. The magnetization along the $z$-direction, $\langle S^z\rangle $, is different from zero for any value of $h$, but has a singular behavior in its first derivative at the transition point.   In the whole interval $0<\gamma\le 1$ the transition belongs to the Ising universality class. For $\gamma=0$ the QPT is of the  Berezinskii-Kosterlitz-Thouless type. The ET transition, occurring at the
the factorization field is also associated with a change in the way the correlations decay with distance: for
$h<h_f$, the decay is oscillatory while, for $h>h_f$, the decay is monotonic \cite{XYsol}. 
The quantum XY model can be experimentally realized with the magnetic compound ${\rm Co \, Nb_2 \, O_6}$~\cite{coldea}. 
\begin{figure}
\begin{center}
\includegraphics[scale=0.43]{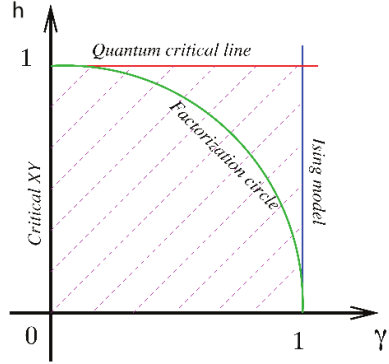} 
\end{center}
\caption{The zero temperature phase diagram of the one dimensional anisotropic $XY$ model in transverse field. Along the quantum critical line the model identifies the Ising universality class with indices $z=\nu=1$; in the hatched area (with $\gamma>0$) the system display long range order in the $x-y$ spin components. The critical $XY$ regime coicide with that one of the XXZ model for $\Delta=0$.
The factorization of the ground state occurs along the circle $h=\sqrt{1-\gamma^2}$.}
\label{f3-2}
\end{figure}

\paragraph{(III)}
The antiferromagnetic anisotropic $XXZ$ chain can be obtained from Eq.~(\ref{general-spin}) 
by setting , $J_x = J_y \equiv J$, and $J_z = \Delta$. 
At zero field it presents a critical $xy$ phase with quasi-long range order (quasi-lro)
for $\vert J_z \vert < 1$; this is separated by two classical phases with QPTs 
at $J_z = \pm 1$. For $h\neq 0$ the $xy$ phase is a strip in the phase diagram, 
eventually turning into polarized phases for sufficiently strong magnetic field. Here we remark that the factorization phenomenon degenerates in the saturation occurring as a first order transition.
Despite the similarities the factorization and  saturation are quite distinct phenomena.
In the case in which $h/J=0$ and for any value of $\Delta$ the model is referred to as the $XXZ$ model. The two isotropic points $\Delta=1$ 
and $\Delta=-1$ describe the antiferromagnetic and ferromagnetic chains respectively (see the phase diagram in Fig.~\ref{f3-3}~\cite{Takahashi-book}). The 
one-dimensional $XXZ$ model can be solved exactly by Bethe Ansatz technique (see e.g.~\cite{Takahashi-book}) and the  correlation 
functions can be expressed in terms of certain determinants (see~\cite{Korepin-book} for a review).  Correlation functions, especially for intermediate 
distances, are in general  difficult to evaluate, although important steps in this direction have been made~\cite{Mallet,Korepin-inverse}. 

\begin{figure}
\begin{center}
\includegraphics[scale=0.35]{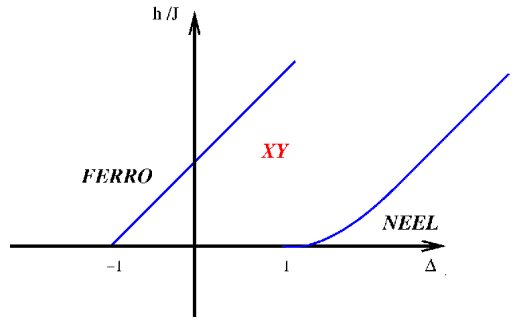} \end{center}
\caption{Zero temperature phase diagram of the spin $1/2$ XXZ model in one dimension. The $XY$ phase is characterized by power law decay of the $xy$ correlations. The Neel and the XY phases are separated by a line of second order phase transitions; at $\Delta=1$ the transition is of Berezinskii-Kosterlitz-Thouless. The ferromagnetic and $XY$ phases are separated by first order phase transitions due to simple level crossings; the onset to the ferromagnetic phase occurs through  the saturation phenomenon.}
\label{f3-3}
\end{figure}

The zero temperature phase diagram of the $XXZ$ model at zero magnetic field  shows a gapless phase in the interval $-1\le \Delta < 1$ with power 
law decaying correlation functions~\cite{Mikeska_xx,Tonegawa81}.  Outside this interval the excitations are gapped. The two phases are separated by a 
Berezinskii-Kosterlitz-Thouless  phase transition at $\Delta=1$ while at $\Delta=-1$  the transition is of the first order.  In the presence of the  external magnetic field a finite  energy gap appears in the spectrum.

\subsection{The reduced density matrix}
In order to compute the QD $ Q_r $ between any two spins $ A $ and $ B $
at distance $ r $ along the chain, the key ingredients are 
the single-site density matrices $ \hat{\rho}^A,\hat{\rho}^B $ and
the two-site density matrix of the composite subsystem $ \hat{\rho}^{AB} $ [see Eq.\eqref{qcorrel}].
Due to translational invariance along the chain, single-site density matrices 
are the same for any spin and they are given by
\begin{equation}
\label{1site_rho}
\hat{\rho}^A = \hat{\rho}^B = \frac{1}{2} \left(\begin{array}{cc} 1 + g_z & g_x \\ g_x & 1 - g_z \end{array} \right) \, ,
\end{equation}
where $g_\alpha \equiv \langle \hat{\sigma}_\alpha \rangle$ are the local expectation values of the magnetization
along the three different axes.

The two-site reduced density matrix for a Hamiltonian model Eq.(\ref{general-spin}) reads

\begin{equation}
  \label{rhoAB}
  \hat{\rho}_r= \frac{1}{4}
  \left(\begin{array}{cccc}\text{A} & \text{a} & \text{a} & \text{F} \\ \text{a} & \text{B} & \text{C} & \text{b} \\ \text{a} & \text{C} & \text{B} & \text{b} \\ \text{F} & \text{b} & \text{b} & \text{D}\end{array}\right)
\end{equation}
in the basis $ \{ \ket{00} , \ket{01} , \ket{10} , \ket{11} \} $, 
where $ \ket{0} $ and $ \ket{1} $ are eigenstates of $ \hat{\sigma}^z $
(because of translational invariance, this density matrix depends only on the distance $ r $ 
between the two spins: $ \hat{\rho}_r \equiv \hat{\rho}^{AB}$). 
The various entries in Eq.~\eqref{rhoAB} are related to the two-point correlators 
$g_{\alpha \beta}(r) \equiv \langle \hat{\sigma}^\alpha_j \hat{\sigma}^ \beta_{j+r} \rangle$ and 
to the local magnetizations, according to the following:
\begin{equation}
  \label{rho_elements_parity}
  \begin{split}
    \text{A}& = 1+g_{z}+g_{zz} , \\
    \text{D}& = 1-g_{z}+g_{zz} , \\
    \text{B}& = 1-g_{zz}        \\
    \text{C}& = g_{xx}+g_{yy} ,  \\
    \text{F}& = g_{xx}-g_{yy} .
  \end{split}
\end{equation}
The entries above express the parity coefficients, while
\begin{equation}
  \label{rho_elements_symbr}
  \begin{split}
    \text{a}& = g_{x}+g_{xz},\\
    \text{b}& = g_{x}-g_{xz} \, .
  \end{split}
\end{equation}
As long as the state displays the $ Z_2$-symmetry of the Hamiltonian (i.e., the thermal ground state) the 
matrix elements $ \text{a} $ and $ \text{b}$ vanish. In this case, the reduced density matrix correspond to an 
$X$ state. On the other hand, the requirement $ \text{a,b} \neq 0$ characterizes the state with symmetry 
breaking~\cite{Osterloh,Oliveira}.

\section{Quantum discord in thermal ground states}

\label{section_thermal}

In this section we summarize the properties of QD in the thermal ground states of spin models with $Z_2$-symmetry. In this case the density matrices are given by requiring $a=b=0$ into Eq.~(\ref{rhoAB}). 

\subsection{XY model}
We begin with the  transverse field XY chain (case (I) of Sec.~\ref{section_models}).
Here, we will be interested in the correlations between arbitrary distant spin pairs (not only nearest-nighbor pairs). 
Then, the reduced state $\rho_{0r}$ for spin pairs at a distance $r=|i-j|$ reads
\begin{equation}
\rho_{0r}=\frac{1}{4}\{\mathbf{I}^{0n}+\langle\sigma^{z}\rangle(\sigma_{0}^{z}+\sigma_{r}^{z}) 
+\sum_{i={x,y,z}}\langle\sigma_{0}^{i}\sigma_{r}^{i}\rangle\sigma_{0}^{i}\sigma_{r}^{i}\},
\label{RDO}
\end{equation}
where $\mathbf{I}^{0n}$ is the identity operator acting on the joint state space of the spins $0$ and $n$. The magnetization density $\langle\sigma^{z}\rangle$ as well as the two-point correlation functions $\langle\sigma_{0}^{i}\sigma_{r}^{i}\rangle$ can be directly obtained from the exact solution of the model~\cite{XYsol}. Then, the total information shared by the spins in the state (\ref{RDO}) is given by 
\begin{equation}
I(\rho_{0r})=S(\rho_{0})+S(\rho_{r})-S(\rho_{r}), 
\label{IXY}
\end{equation}
with 
\begin{equation}
S(\rho_{0})=S(\rho_{r})= -\sum_{i=0}^{1}\frac{1+(-1)^{i}\langle\sigma^{z}\rangle}{2}\log_{2}
\frac{1+(-1)^{i}\langle\sigma^{z}\rangle}{2}
\end{equation}
and
\begin{equation}
S(\rho_{0r})=\sum_{i=0}^{1}(\xi_{i}\log_{2}\xi_{i}+\eta_{i}\log_{2}\eta_{i}), 
\end{equation}
where
\begin{eqnarray}
\xi_{i}&=&\frac{1+\langle\sigma_{0}^{z}\sigma_{r}^{z}\rangle}{4}+\frac{(-1)^{i}}{4}\sqrt{(\langle\sigma_{0}^{x}\sigma_{r}^{x}\rangle-\langle\sigma_{0}^{y}\sigma_{r}^{y}\rangle)^{2} +4\langle\sigma^{z}\rangle^{2}}, \nonumber \\
\eta_{i}&=&\frac{1}{4}\left[1-\langle\sigma_{0}^{z}\sigma_{r}^{z}\rangle+(-1)^{i}(\langle\sigma_{0}^{x}\sigma_{r}^{x}\rangle+\langle\sigma_{0}^{y}\sigma_{r}^{y}\rangle)\right].
\end{eqnarray}
Following \cite{Sarandy:09,Maziero:12,Maziero:10}, we obtain that the minimum in Eq. (\ref{qcorrel}) is attained, for all values of $h$, $\gamma$, and $n$, by the following set of projectors: $\{|+\rangle\langle+|,|-\rangle\langle-|\}$, with $|\pm\rangle=(\left\vert\uparrow\right\rangle \pm\left\vert \downarrow\right\rangle )/\sqrt{2}$, where $\{\left\vert \uparrow\right\rangle ,\left\vert \downarrow\right\rangle \}$
are the eigenstates of $\sigma^{z}$. Thus one obtains
\begin{equation}
J(\rho_{r})=H_{bin}(p_{1})+H_{bin}(p_{2}),
\end{equation}
where $H_{bin}(x)$ is the binary entropy 
\begin{equation}
H_{bin}(x)=-x\log_{2}x-(1-x)\log_{2}(1-x)
\end{equation}
and
\begin{equation}
p_{1}=\frac{1}{2}\left(1+\langle\sigma^{z}\rangle\right), \hspace{0.3cm} 
p_{2}=\frac{1}{2}\left(1+\sqrt{\langle\sigma_{0}^{x}\sigma_{r}^{x}\rangle^{2}+\langle\sigma^{z}\rangle^{2}}\right).
\label{finalXY}
\end{equation} 
This provides therefore an analytical expression for evaluating QD. We 
can then use this expression to investigate QD at quantum criticality. 
By fixing $\gamma=1$ (Ising model) and considering pairwise QD between first 
neighbors, we can identify the ferromagnetic-paramagnetic QPT by looking at 
either classical or quantum correlations, as shown in Figs.~\ref{f4-2} and~\ref{f4-3}. 

\begin{figure}[ht]
\centerline{\psfig{file=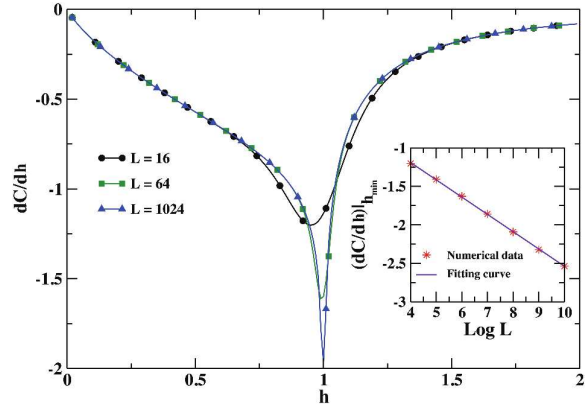,width=3in}}
\vspace*{8pt}
\caption{(Color online) First derivative of the classical correlation for nearest-neighbor 
spins with respect to $h$  in the transverse field Ising chain for different lattice 
sizes $L$. The derivative of $C$ has a pronounced minimum at $h_{min}$, which tends to the critical 
point $h=1$ as $L\rightarrow \infty$. Inset: $dC/dh$ taken at $h_{min}$ exhibits a logarithmic divergence 
fitted by $\left.(dC/dh)\right|_{h_{min}} = -0.29161 - 0.22471 \log L$. [Original plot from Ref.~11]}
\label{f4-2}
\end{figure}

\begin{figure}[ht]
\centerline{\psfig{file=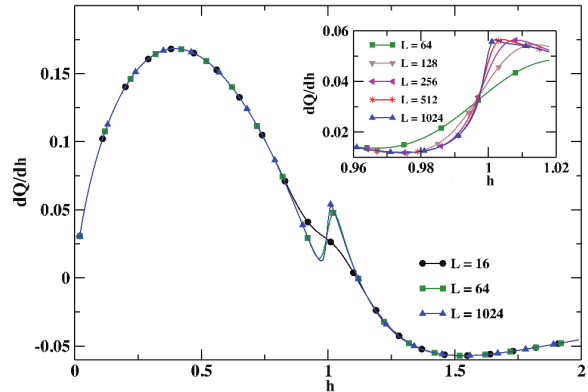,width=3in}}
\vspace*{8pt}
\caption{(Color online) First derivative of the quantum correlation for nearest-neighbor 
spins with respect to $h$ in the transverse field Ising chain for different lattice sizes $L$. 
Inset: $dQ/dh$ presents an inflection point that tends to the QCP $h=1$ as $L\rightarrow \infty$. 
[Original plot from Ref.~11]}
\label{f4-3}
\end{figure}

As easily observed, derivatives of both 
$C(\rho)$ and $Q(\rho)$ exhibit a non-analytical behavior at the quantum critical 
point as we approach to the limit of an infinite system. 

\begin{figure}[ht]
\centerline{{{\psfig{file=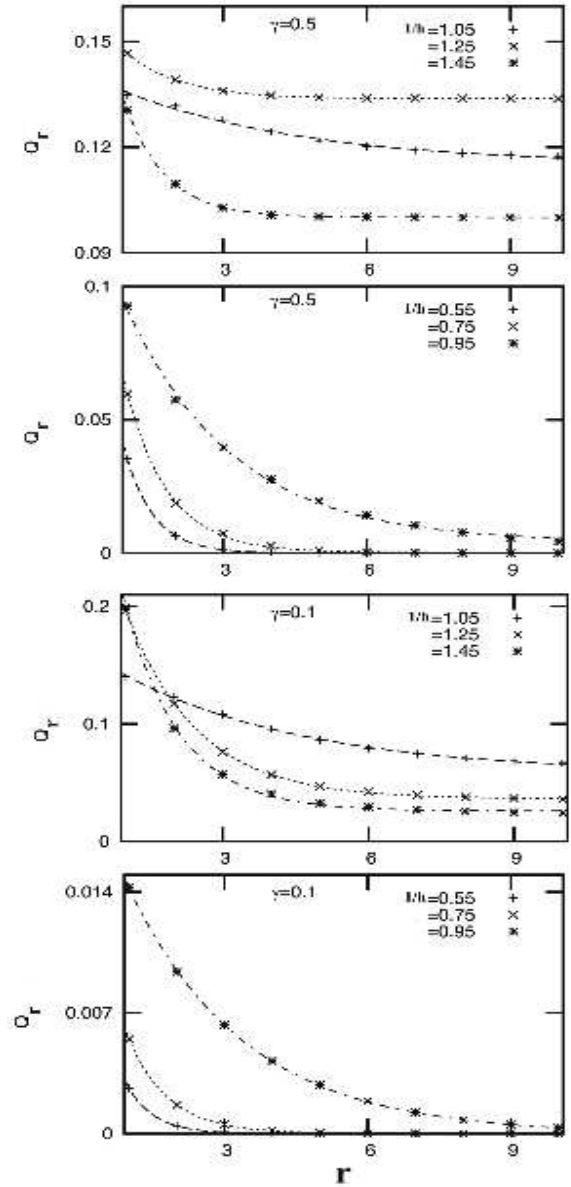,height=16cm,width=7.8cm}}}}
\vspace*{8pt}
\caption{Decay of quantum discord with the distance between the spins sites for some values of $h$ and $\gamma$. 
The points are the computed values of QD and the lines are the exponential fits (see the text for details). 
[Original plot from Ref.~24]}
\label{f4-4}
\end{figure}

Remarkably, besides the non-analytical behavior of pairwise QD derivatives, we can show that QD exhibits a 
long-range decay~\cite{Maziero:12}. This is in contrast with the two-spin entanglement behavior, which is 
typically very short-ranged. In order to illustrate the long-range behavior of QD, we plot in Fig.~\ref{f4-4} 
pairwise QD as a function of distance $n$. The curves in Fig.~\ref{f4-4} are the exponential fits of QD. We 
observe that, for both examples of anisotropies considered, $\gamma=0.1$ and $\gamma=0.5$, the decay of QD 
with distance can be well fitted by an exponential function $a+b\exp(-c\, r)$, where $a$, $b$, and $c$ are 
constants. Nevertheless, we notice that, while for $h>1$ QD vanish exponentially, in the cases where $h<1$, 
we obtain a constant long-distance value for QD that depends only on $\gamma$ and $h$.  
It is also remarkable to observe that the factorization phenomenon can be traced already for the thermal 
ground state: it is the unique value of the field where the same quantum correlations are present at any 
length scale (left inset of Fig.~\ref{qd}).  This is due by the peculiar pattern of correlation functions 
close to the factorizing field in Eq.(\ref{corr_pattern}).

\subsection{XXZ model}
As a second example, let us consider now $XXZ$ spin-1/2 chain with  $h=0$ (the case (III) of Sec.\ref{section_models}).
In this case the  Hamiltonian exhibits $U(1)$ invariance, which ensures that the element $F$ 
of the reduced density matrix given by Eq.~(\ref{rhoAB}) vanishes. Moreover, the ground state has magnetization density $g^k_{z} = \langle \sigma_z^k \rangle=0$ ($\forall\, k$), which implies that $A = D = (1/4) \left(1+g_{zz}\right)$. For convenience, we consider correlations between nearest-neighbor spin pairs, i.e. $\rho_r = \rho_1 \equiv \rho$, and then expand $\rho$ in terms of Pauli operators, which reads
\begin{equation}
\rho = \frac{1}{4} \left[ I\otimes I + \sum_{i=1}^{3} \left( c_i \sigma^i \otimes \sigma^i \right) \right], 
\end{equation}
with 
\begin{eqnarray}
c_1 = c_2 = \frac{1}{2} \left( g_{xx} + g_{yy} \right), \hspace{0.5cm} c_3 = g_{zz} .
\label{gen-c}
\end{eqnarray} 
By taking into account those notations, classical correlations in Eq.~(\ref{ccorrel}) can be analytically worked out~\cite{Luo:08}, yielding  
\begin{equation}
C(\rho) = \frac{\left(1-c\right)}{2} \log \left(1-c\right) + \frac{\left(1+c\right)}{2} \log \left(1+c\right),
\label{cc-xxz}
\end{equation}
with $c = \max\left(|c_1|,|c_2|,|c_3|\right)$. For the mutual information $I(\rho)$ we obtain
\begin{equation}
I(\rho) = 2 + \sum_{i=0}^{3} \lambda_i \log \lambda_i, 
\end{equation}
where
\begin{eqnarray}
\lambda_0 &=& \frac{1}{4} \left( 1-c_1-c_2-c_3 \right), \nonumber \\
\lambda_1 &=& \lambda_2 = \frac{1}{4} \left( 1 + c_3 \right), \nonumber \\
\lambda_3 &=& \frac{1}{4} \left( 1+ c_1 + c_2 - c_3 \right).
\label{qmi-ev-xxz}
\end{eqnarray}
In order to compute $C(\rho)$ and $Q(\rho)$ we write $c_1$, $c_2$, and $c_3$ in terms of the ground state 
energy density. By using the Hellmann-Feynman theorem~\cite{Hellmann:37,Feynman:39} for the XXZ  Hamiltonian $H_{XXZ}$, we obtain
\begin{eqnarray}
c_1 &=& c_2 = \frac{1}{2} \left(G_{xx} + G_{yy}\right) = \Delta \frac{\partial \varepsilon_{xxz}}{\partial \Delta} 
- \varepsilon_{xxz} \, , \nonumber \\
c_3 &=& G_{zz} = -2 \frac{\partial \varepsilon_{xxz}}{\partial \Delta} \, ,
\label{c-xxz}
\end{eqnarray}
where $\varepsilon_{xxz}$ is the ground state energy density 
\begin{equation}
\varepsilon_{xxz} = \frac{\langle \psi_0| H_{XXZ} |\psi_0 \rangle}{L} = - \frac{1}{2} \left(G_{xx} + 
G_{yy} + \Delta G_{zz} \right),
\label{aux-xxz}
\end{equation}
with $|\psi_0\rangle$ denoting the ground state of $H_{XXZ}$. Eqs.~(\ref{c-xxz}) and~(\ref{aux-xxz}) hold for a chain with an arbitrary number of sites, allowing the discussion of correlations either for finite or infinite chains. Indeed, ground state energy as well as its derivatives can 
be exactly determined by Bethe Ansatz technique~\cite{Yang:66}, which allows us to obtain the 
correlation functions $c_1$, $c_2$, and $c_3$. In Fig.~\ref{f4-1}, we plot classical and 
quantum correlations between nearest-neighbor pairs for an infinite XXZ spin chain. 

\begin{figure}[th]
\centering {\includegraphics[angle=0,scale=0.33]{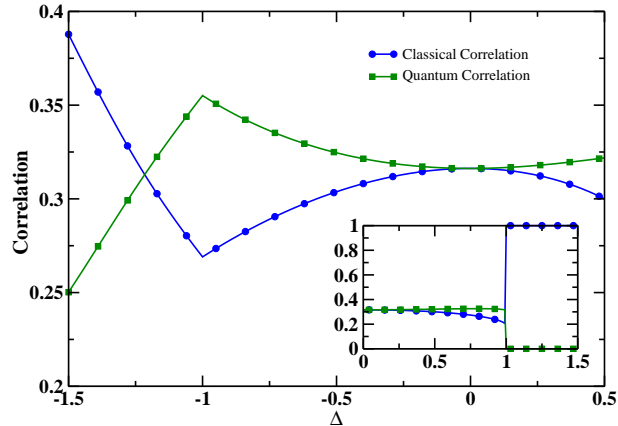}}
\caption{(Color online) Quantum and classical correlations for nearest-neighbor spins in the XXZ chain for 
$L\rightarrow \infty$.  [Original plot from Ref.~11]}
\label{f4-1}
\end{figure}

Note that, in the classical Ising limit $\Delta \rightarrow \infty$, we have a  fully polarized ferromagnet. The ground state is then a doublet given by the vectors 
$|\uparrow \uparrow \cdots \uparrow\rangle$ and $|\downarrow \downarrow \cdots \downarrow\rangle$, 
yielding the mixed state
\begin{equation}
\rho = \frac{1}{2} |\uparrow \uparrow \cdots \uparrow\rangle \langle \uparrow \uparrow \cdots \uparrow| 
+ \frac{1}{2} |\downarrow \downarrow \cdots \downarrow\rangle \langle \downarrow\downarrow \cdots \downarrow|. 
\label{xxz-bpd}
\end{equation}
Indeed, this is simply a classical probability mixing, with $C(\rho) = I(\rho) = 1$ and $Q(\rho) = 0$. 
The same applies for the antiferromagnetic Ising limit $\Delta \rightarrow - \infty$, where a doubly degenerate 
ground state arises. Moreover, observe that the classical (quantum) correlation is a minimum (maximum) at the 
infinite-order QCP $\Delta=-1$. On the other hand, both correlations are discontinuous at the first-order QCP 
$\Delta=1$. The behavior of QD at the first-order QPT is in agreement with the expected behavior for   
entanglement, as shown by Wu {\it et al.} in Ref. \refcite{Wu:04}. In fact, the approach used by Wu {\it et al.} 
for entanglement also applies for any function of the reduced density matrix of two-spins. It is only based 
on the fact that the non-analyticities in the derivatives of the ground state energy typically 
come from the elements of the reduced density matrix and, as such, may be inherited by entanglement or  
other functions of the density matrix elements. There is a caveat in the fact that entanglement and QD 
definitions involve a maximization or minimization processes, which may create accidental non-analyticities or
even hide ones. In fact, the behavior of the entanglement as measured by concurrence~\cite{Wootters:98} or 
negativity~\cite{Vidal:02a} at $\Delta=1$, is not the one expected for first-order QPTs, but rather for second order 
QPTs, since it is not discontinuous at $\Delta=1$ but its derivative is. Such a misidentification of the order of 
the QPT is due to the maximization process involved~\cite{Yang:05}. Remarkably, although involving an optimization, 
QD {\it does} indicate the correct order of the QPT, exhibiting in this case a behavior that is superior to entanglement. 
For $\Delta=-1$, the result for QD is in agreement with the entanglement behavior (maximum or minimum at the infinity-order QPT), 
even though there is no proof of the generality of this behavior for arbitrary infinite-order QPTs.

\section{Quantum discord in ground states with symmetry breaking}
\label{section_broken}

In this section, we discuss the main features of QD for $\rho_r$ given by Eq.(\ref{rhoAB}) with non vanishing $a,b$ 
given by (\ref{rho_elements_symbr}). This correspond to the ground state of the Hamiltonian (\ref{general-spin}) 
(with nearest neighbor couplings), with broken symmetry.  In this case the reduced density matrix is accessed via 
DMRG  for finite systems with open boundaries~\cite{dmrg}, by adding a small symmetry-breaking  longitudinal field 
(of typical magnitude $h_x = 10^{-6}$) to the Hamiltonian~\footnote{ For the quantum Ising model  the reduced density 
operator with $a,b \neq0$ could be accessed analytically, in principle. Nevertheless, we resort to DMRG because the 
$g_{xz}(r)$ is given in terms of integral contour formulas~\cite{johnson}.}. In  Fig.~\ref{qd}, it is displayed QD 
for the  quantum XY model.  Similarly, QD for XYX model in the ground state with symmetry breaking is obtained in Fig.(\ref{qdxyx}). 

\begin{figure}[ht]
  \centerline{\psfig{file=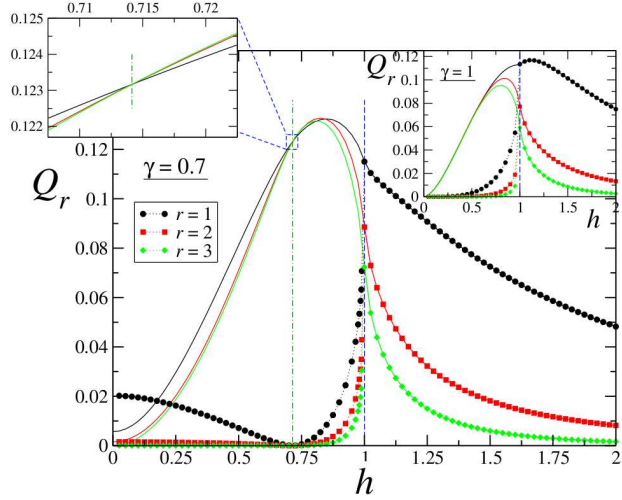,width=3.2in}}
  \caption{Quantum discord $Q_r(h)$ between two spins at distance $r$ in the $XY$ model 
    at $\gamma = 0.7$ (main plot and left inset) and $\gamma = 1$ (right inset), 
    as a function of the field $h$.
    Continuous lines are for the thermal ground state, while symbols denote the
    symmetry-broken state obtained  with DMRG
    in a chain of $L=400$ spins; simulations were performed by keeping $m=500$ states
    and evaluating correlators at the center of the open-bounded chain. 
    For $\gamma=0.7$ and at $h_f \simeq 0.714$, in the symmetric state all the curves 
    for different values of $r$ intersect, while after breaking the symmetry
    $Q_r$ is rigorously zero.
    At the critical point $Q_r$ is non analytic, thus signaling the QPT. 
    In the paramagnetic phase, there is no symmetry breaking to affect $Q_r$. [Original plot from Ref.~36]}
  \label{qd}
\end{figure} 


\begin{figure}[ht]
  \centerline{\psfig{file=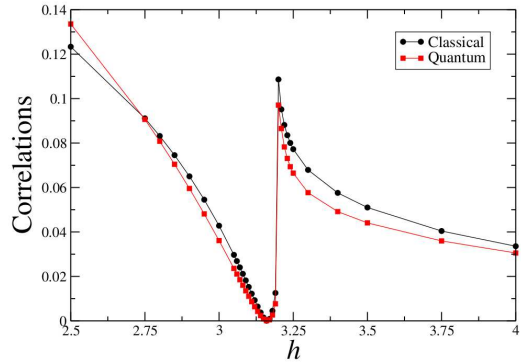,width=3.2in}} 
  \caption{Quantum discord $Q_r(h)$ between two spins at distance $r$ in the $XYX$ model.}
  \label{qdxyx}
\end{figure} 


QD is substantially affected by SSB~\footnote{In contrast, bipartite entanglement is not affected by SSB around 
the quantum critical point. SSB slightly influences bipartite entanglement only for $h<h_f$. However, multipartite 
entanglement is highly affected~\cite {Oliveira}.}. 
We observe that quantum correlations are typically much smaller
deep in the ordered ferromagnetic phase $h < h_c$, rather than in the paramagnetic one $h > h_c$. 
Nonetheless, as we shall see, they play a fundamental role to drive the order-disorder
transition at the QPT, where $Q_r$ exhibits a maximum, as well as the correlation
transition at $h_f$, where $Q_r$ is rigorously zero. 
The QPT is marked by a divergent 
derivative of the QD (see also~\cite{Dillenschneider:08,Sarandy:09,Maziero:10}). 
Such divergence is present at every $\gamma$, for the symmetry broken state;
on the other hand, for the thermal ground state, it is not present at $\gamma=1$. 
A thorough finite-size scaling analysis is shown in Fig.~(\ref{scaling}(b)) proving that $z=\nu=1$.  
For the thermal ground state (in the thermodynamic limit), we found that $\partial_h Q_r$ 
diverge logarithmically as $\partial_h Q_r \sim  \ln |h-h_c| $, within the Ising universality class.

\begin{figure}[!t]
  \centering
(a) \subfigure[] {\includegraphics[width=0.48\columnwidth]{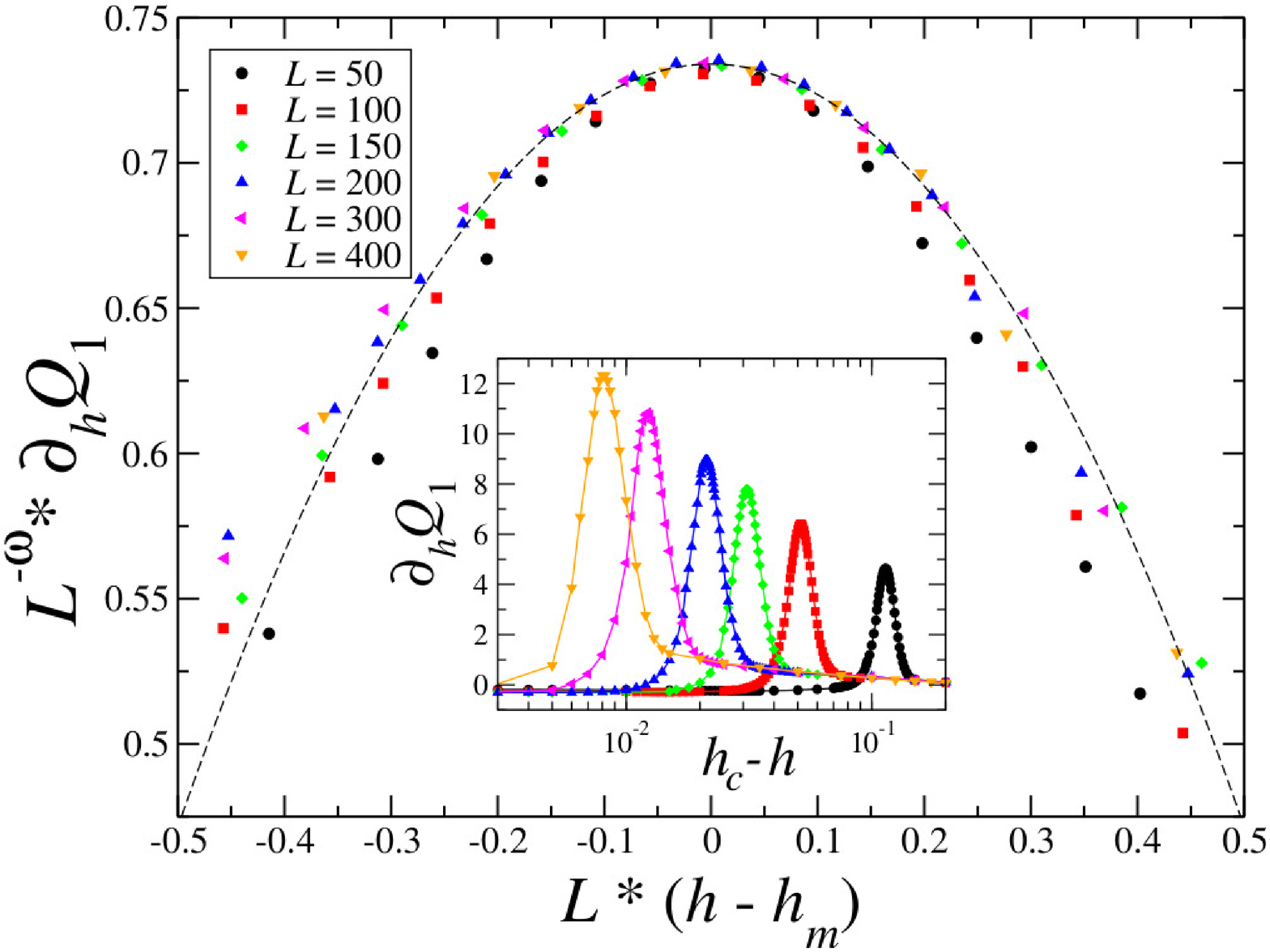}}
 (b)\subfigure[] { \includegraphics[width=0.48\columnwidth]{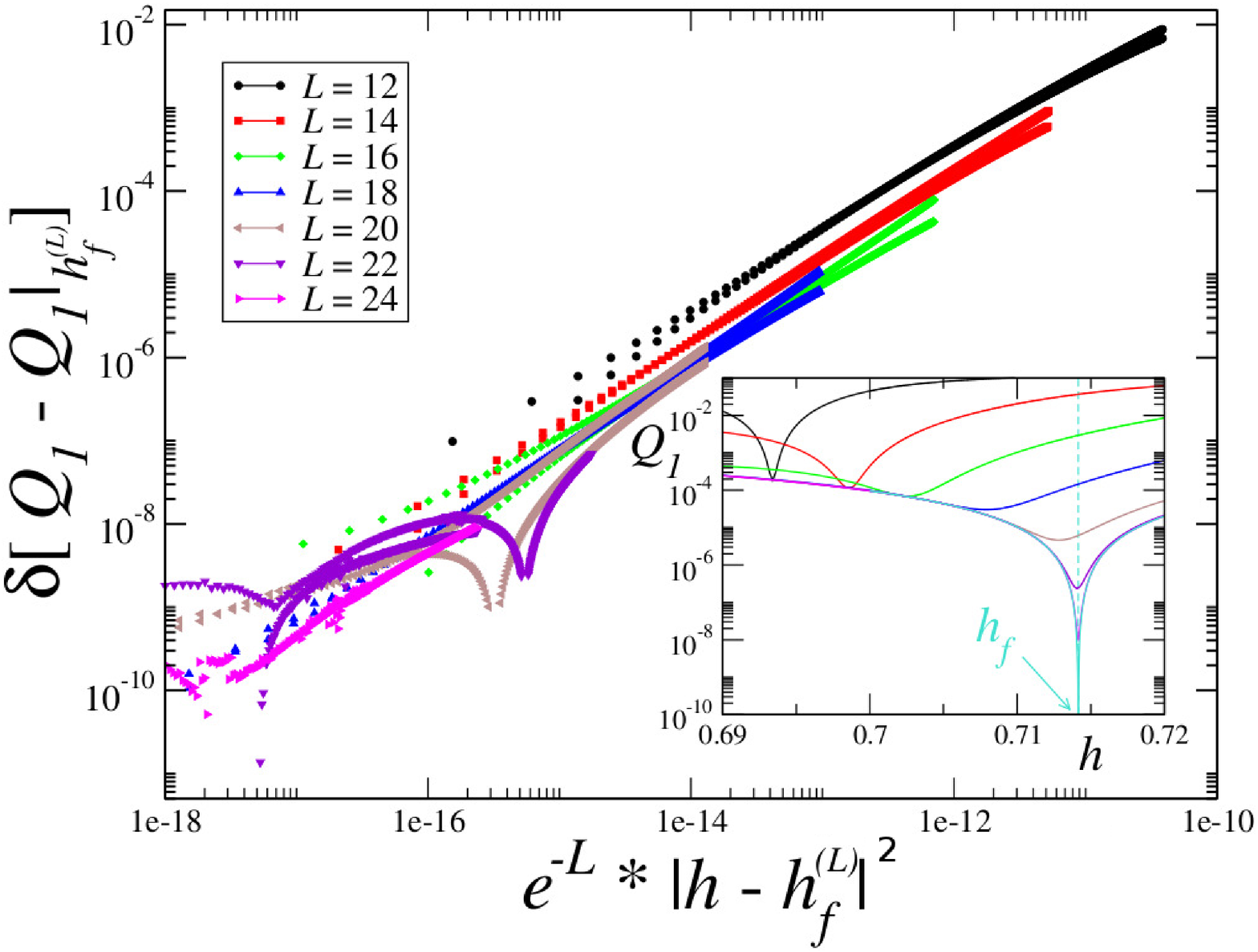}}
  \caption{(a) Finite-size scaling of $\partial_h Q_1$ for the
    symmetry-broken state in proximity of the critical point $h_c$. 
    Displayed data are for $\gamma= 0.7$. 
    The first derivative of the QD is a function of $L^{-\nu}(h-h_m)$ only, 
    and satisfies the scaling ansatz $\partial_h Q_1 \sim L^\omega \times F[L^{-\nu}(h-h_m)]$,
    where 
    $h_m$ is the renormalized critical point at finite size $L$ and $\omega=0.472$. 
    We found a universal behavior $h_c - h_m \sim L^{-1.28 \pm 0.03}$ 
    with respect to $\gamma$. 
    Inset: raw data of $\partial_h Q_1$ as a function of the transverse field.
    (b) Scaling of $Q_1$ close to the factorizing field, for $\gamma = 0.7$:
    we found an exponential convergence to the thermodynamic limit, 
    with a universal behavior according to $e^{-\alpha L}(h-h_f^{(L)})$,
    $\alpha \approx 1$ [$h_f^{(L)}$ denotes the effective factorizing field 
    at size $L$, while $\delta(Q_1) \equiv Q_1^{(L)} - Q_1^{(L \to \infty)}$].
    Due to the extremely fast convergence to the asymptotic value,
    already at $L \sim 20$ differences with the thermodynamic limit
    are comparable with DMRG accuracy.
    Inset: raw data of $Q_1$ as a function of $h$. The cyan line is for $L = 30$
    so that, up to numerical precision, the system behaves at the thermodynamic limit. [Original plot from Ref.~36]
    }
  \label{scaling}
\end{figure}

%

At the factorizing field $h_f$, all the correlation measures are zero in the state 
with broken symmetry (see symbols in Fig.~\ref{qd}); in particular, we numerically found 
a dependence 
\begin{equation}
Q_r \sim (h-h_f)^2 \times \big( \frac{1-\gamma}{1+\gamma} \big)^r
\end{equation}
 close to it. 
Such behavior is consistent with the expression of correlation functions 
close to the factorizing line Eq.(\ref{corr_pattern}),  and here appears 
incorporating the effect arising from the non vanishing spontaneous magnetization. 

It is found that  $Q_r$ exponentially tends to the asymptotic value $Q_r^{(L \to \infty)}$ (see Fig.~\ref{scaling}(b)). 
In~\cite{Amico:06,fubini}, it was shown that $h_f$ marks the transition between 
two different patterns of entanglement. The factorization is thus a new kind of zero-temperature 
transition of collective nature, not accompanied by a change of symmetry. 
We emphasize, though, that this transition does not correspond to any non analyticity in the ground state as a function of $h$.
\footnote{Accordingly, the fidelity  $\mathcal F(h)\equiv |\langle gs(h)|gs(h+\delta h\rangle|$  (which can detect both symmetry breaking and non-symmetry breaking QPTs), 
is a smooth function at $h_f$.}. 
At the level of spectral properties of the system, we interpret this result 
as an effect of certain {\it competition between} states belonging to different parity sectors 
for finite $L$~\cite{giorgidepasquale}; as these states intersect, the 
ground-state energy density is diverging {\it for all finite $L$} (such divergence, though, 
vanishes in the thermodynamic limit). 


\section{Witnessing quantum correlated states}
\label{section_witness}

Computation of QD involves an extremization procedure. However, if we just want to find out 
whether a state exhibits quantum correlations, such a procedure can be avoided by introducing 
witness operators, namely, observables able to detect the presence of nonclassical states. 
A witness $\cal{W}$ is  defined here as a Hermitian operator whose norm $\| {\cal W} \|$ is such 
that, for all nonclassical states, $\| {\cal W} \| > 0$, with a convenient norm measure adopted. 
Therefore, $\| {\cal W} \| > 0$ is a sufficient condition for nonclassicality (or, equivalently, 
a necessary condition for classicality).  
Let us begin by defining a classical state. We will concetrate in von Neumann meansurements, but our 
approach can be generalized for POVMs. 

\begin{definition}
If there exists any measurement $\{\Pi_j\}$ such that $\Phi(\rho)=\rho$ then $\rho$ describes a {\it classical} state 
under von Neumann local measurements.
\end{definition}

Therefore, it is always possible to find out a local measurement basis such that a classical state 
$\rho$ is kept undisturbed. In this case, we will denote $\rho \in {\cal C}^N$, where ${\cal C}^N$ 
is the set of $N$-partite classical states. Observe that such a definition of classicality coincides 
with the vanishing of global QD~\cite{Rulli:11}. Note also that $\Phi(\rho)$ by itself 
is a classical state for any $\rho$, namely, $\Phi(\Phi(\rho)) = \Phi(\rho)$. Hence, $\Phi(\rho)$ can 
be interpreted as a decohered version of $\rho$ induced by measurement. A witness for nonclassical states 
can be directly obtained from the observation that the elements of the set $\{\Pi_j\}$ are eigenprojectors of $\rho$. 
This has been shown in Ref.~\refcite{Saguia:11}, being stated by Theorem 2 below (see also Ref.~\refcite{Luo:08}).

\begin{theorem}
\,\,$\rho \in {\cal C}^N \Longleftrightarrow \left[ \rho, \Pi_j \right] = 0 \,\,\,(\forall j)$, with 
$\Pi_j = \Pi_{A_1}^{i_1} \otimes \cdots \Pi_{A_N}^{i_N}$ and $j$ denoting the index string 
$(i_1 \cdots i_N)$.
\label{t1}
\end{theorem}
\begin{proof}
If $\rho \in {\cal C}^N$ then $\Phi(\rho)=\rho$. Then, similarly as in Ref.~\refcite{Luo:08}, a direct 
evaluation of $\sum_j \left[\rho,\Pi_j\right]\,\left[\rho,\Pi_j\right]^{\dagger}$ yields
$\sum_j \left[\rho,\Pi_j\right]\,\left[\rho,\Pi_j\right]^{\dagger} = \Phi(\rho^2)-\rho^2$. 
However, $\Phi(\rho)=\rho$ also implies that $\Phi(\rho^2)=\rho^2$. Therefore, 
\begin{equation}
\sum_j \left[\rho,\Pi_j\right]\,\left[\rho,\Pi_j\right]^{\dagger} = 0\,.
\end{equation}
Hence, $\rho \in {\cal C}^N \Longrightarrow \left[\rho,\Pi_j\right]=0$. 
On the other hand, if $\left[\rho,\Pi_j\right]=0$ then $\{\Pi_j\}$ provides a 
basis of eigenprojectors of $\rho$. Then, from the spectral decomposition, we obtain
\begin{equation} 
\rho = \sum_j p_j \Pi_j = \sum_{i_1,\cdots,i_N} p_{i_1,\cdots,i_N} \Pi_{A_1}^{i_1} \otimes \cdots \Pi_{A_N}^{i_N}\, ,
\label{rho-sd}
\end{equation}
which immediately  implies that $\Phi(\rho)=\rho$. Hence $\left[\rho,\Pi_j\right]=0 \Longrightarrow \rho \in {\cal C}^N$.
\end{proof}

We can now propose a necessary condition to be obeyed for arbitrary multipartite classical states.

\begin{theorem}
Let $\rho$ be a classical state and $\rho_{A_i}$ the reduced density operator for the subsystem $A_i$. 
Then $\left[\rho,\rho_{A_1} \otimes \cdots \otimes \rho_{A_N}\right] = 0$.
\end{theorem}
\begin{proof}
From theorem~\ref{t1}, if $\rho \in {\cal C}^N$ then the spectral decomposition of $\rho$ yields Eq.~(\ref{rho-sd}). 
Therefore, 
\begin{equation}
\rho_{A_i} = \sum_{r} p_n \Pi_{A_i}^{n} \,.
\end{equation}
Hence, by direct evaluation, we obtain that $\left[\rho,\rho_{A_1} \otimes \cdots \otimes \rho_{A_N}\right] = 0$.
\end{proof}

Observe that, given a composite multipartite state $\rho$, it is rather simple to evaluate the commutator 
$\left[\rho,\rho_{A_1} \otimes \cdots \otimes \rho_{A_N}\right]$, with no extremization procedure as usually 
required by QD computation. From the necessary condition for classical states above, we can define 
a witness for nonclassicality by the norm of the operator $\left[\rho,\rho_{A_1} \cdots \rho_{A_N}\right]$. 
Indeed, if $\rho \in {\cal C}^N \Longrightarrow \| {\cal W} \| = 0$, where 
\begin{equation}
{\cal W} = \left[\rho,\rho_{A_1} \otimes \cdots \otimes \rho_{A_N}\right] \, .
\end{equation}
Therefore, $\| {\cal W} \| > 0$ is sufficient for nonclassicality. 

For concreteness, we will take $\| {\cal W} \|$  as defined by the trace norm, namely, 
\begin{equation}
\| {\cal W} \| = {\textrm{Tr}} \sqrt{ {\cal W} {\cal W}^\dagger} \, .
\end{equation}
In order to apply the witness $\| {\cal W} \|$ in the XY model, we will first consider the case of the thermal ($Z_2$-symmetric) 
ground state. For this case, a plot of $\| {\cal W} \|$ as a 
function of $h$ is provided in Fig.~\ref{f6-1} for $\gamma=0.6$. Observe that this plot shows that the only 
possible classical state appears for $h=\infty$. Indeed, this point corresponds to an infinite transverse magnetic field applied, 
which leads the system to a product state, with all spins individually pointing in the $z$ direction. 

However, besides large $h$, the ground state with SSB exhibits a further nontrivial factorization (product state) point at 
$ 
\gamma^2 + h^2 = 1.
$
This point identifies a change in the behavior of the correlation functions decay, which pass from
monotonically to an oscillatory decay~\cite{XYsol}. Besides, it has been realized
that, at this point, the product ground state is two-fold degenerated even for finite systems~\cite{Giampaolo:08,Ciliberti:10,Tomasello:11}.
\begin{figure}[ht]
\centering {\includegraphics[clip,scale=0.3]{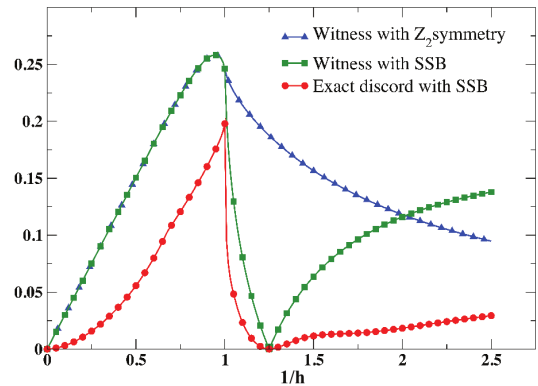}}
\caption{(Color online) Witness $\| {\cal W} \|$ and exact value of QD for a pair of spins 
in both symmetric and broken ground states of the XY spin-1/2 infinite chain as a function of $h$ for $\gamma=0.6$.
[Original plot from Ref.~38]}
\label{f6-1}
\end{figure}

On the other hand, if SSB is taken into account, which must be the case in the thermodynamic limit, this 
picture dramatically changes. Indeed, in the broken case, the two-spin reduced density matrix is more
complicated (see Eqs. (\ref{rhoAB}), (\ref{rho_elements_symbr})).
Tight lower and upper bounds for $a(h,\gamma)$ and $b(h,\gamma)$ are obtained in Ref.~\refcite{Oliveira}. In our plots, 
either of such bounds essentially yields the same curves, which led us to keep the results produced with 
the lower bound. Concerning the results for $\| {\cal W} \|$, we can compare the plots for the thermal ground  
state and the ground state with SSB also in Fig.~\ref{f6-1}. Observe that, in the case with SSB, the only point for 
which $\| {\cal W} \|=0$ is $1/h= 1.25$. Indeed, this corresponds to a classical state, as can 
be confirmed by the exact computation of QD. Since our witness identifies {\it classical-classical} states, 
we adopted the symmetric version of QD~\cite{Maziero-2:10,Rulli:11}. Remarkably, this classical state is associated 
with one of two doubly-degenerated factorized (fully product) ground state, which appears as a consequence of 
the $Z_2$ SSB. This unique classical point for the XY model is then clearly revealed by the witness evaluation. 
Another aspect of the witness observed from Fig.~\ref{f6-1} is that it displays nonanalyticity in its first derivative 
at the QPT  $h=1$. This phenomenon occurs both for the symmetric and broken ground states, 
which promotes $\| {\cal W} \|$  to a simple tool also to detect QPTs.

\section{Conclusion}
\label{section_conclusions}
We reviewed on the two-spin QD in the ground state of spin chains of the XYZ type in a transverse field. 
Close to QPTs, the anomalies of QD reflect the universality class of the Hamiltonian 
of the system. Therefore we can conclude that the scaling of QD and entanglement is identical with the 
known scaling of correlations of standard many body theory. In contrast, the range of quantum correlations 
at the critical point is non universal: entanglement is typically short-ranged (with scaling laws depending on the 
microscopic details of the system), while QD decays algebraically  (similarly as the standard correlation functions) 
following the universality paradigm~\cite{Maziero:12}. In particular, we note that QD is generically more robust 
than entanglement. This ultimately arises because QD catches quantum 
correlations in separable (mixed) state, where entanglement is vanishing. Therefore, at zero temperature, the 
QD can be sensible  at distances  were entanglement is vanishing. Along similar lines of reasoning, QD can survive at 
temperatures where the states are mixed enough to kill the entanglement (see Ref.~\refcite{Werlang:12} 
in this special issue). The investigation of multipartite measures of quantum correlations~\cite{Rulli:11,Dakic:10} 
as well as their impact in the characterization of QPTs is certainly a future challenge. Entanglement  is believed entering 
in its multipartite form close to the QPT; the two-spin entanglement is then a small amount of the entanglement at disposal, 
ultimately because of the monogamy property ~\cite{Osborne:02,Coffman:00}. On the other hand, it has already been shown that 
QD does not obey the usual monogamy relationship~\cite{Prabhu:12,Giorgi:11,Streltsov:12}. However, a generalized monogamy 
constraint can be established~\cite{Braga:12}. The implications of such monogamy for the behavior of multipartite QD at QPTs 
is object of current research. For studies of multipartite QD in spin models, see Refs.~\refcite{Rulli:11,Campbell:11}. 
Moreover, the application of {\it genuine} multipartite measures of QD (e.g., as defined by Refs.~\refcite{Braga:12,Maziero-2:12,Giorgi:12}) 
is also a further promising topic. 

At the factorizing field, the correlations do not depend on the spatial distance. Therefore all the QDs for different ranges 
crosses at $h=h_f$ (see the inset of Fig.\ref{qd}). The actual value of the discords is finite for thermal ground states and it is vanishing for ground states 
with symmetry breaking. This results since the factorized thermal states are mixed; in the case of  symmetry breaking the 
factorized two-spin state is pure, and therefore QD coincides with entanglement (they are both vanishing at $h_f$). 
Here we remark that QD  is a smooth function  close to the factorizing field (in contrast, entanglement displays certain 
discontinuity), and therefore it  can be analyzed  easly (see Fig.\ref{scaling}b)\cite{Tomasello:11,Tomasello:12}.  In fact the finite size scaling of 
the factorization phenomenon was achieved through the QD and not through the analysis of entanglement. We 
finally comment that the role of multipartite entanglement was demonstrated to be negligible close to the factorizing 
field~\cite{Amico:06}. Therefore the two-spin entanglement developed so far can provide a complete enough scenario of entanglement 
close to $h_f$.

To conclude, we mention that QD is also interesting from the experimental perspective. Pairwise non-classical correlations in 
highly mixed states were experimentally measured in NMR systems~\cite{Auccaise:11}. In that experimental protocol non linear 
witness operators (see section \ref{section_witness}) were employed to simplify  the optimization encoded in the discord. 
Interestingly enough the QD can be easly evaluated for spin cluster. This may be important to measure the 
non-classical correlation via neutron scattering\cite{Yurishchev:11}.

\vspace{-0.25cm}

\section*{Acknowledgements} We thank A. Hamma, D. Rossini, B. Tomasello, J. Maziero, L. C\'eleri, 
R. M. Serra, A. Saguia, and C. C. Rulli for discussions. Financial support from the Brazilian agencies 
CNPq (to M. S. S. and T. R. O.) and FAPERJ (to M. S. S.) is acknowledged. This work was performed as part
of the Brazilian National Institute for Science and Technology of Quantum Information (INCT-IQ). 
The  research of L. A. was supported in part by Perimeter Institute for Theoretical Physics. Research at 
Perimeter Institute is supported by the Government of Canada through Industry Canada and by Province of 
Ontario through the Ministry of Economic Development $\&$ Innovation.

\end{document}